\documentclass[a4paper, 12pt]{article}

\usepackage[sort&compress]{natbib}
\bibpunct{(}{)}{;}{a}{}{,} 

\usepackage{amsthm, amsmath, amssymb, mathrsfs, multirow, url, subfigure}
\usepackage{graphicx} 
\usepackage{ifthen} 
\usepackage{amsfonts}
\usepackage[usenames]{color}
\usepackage{fullpage}
\usepackage{setspace}

\theoremstyle{plain} 
\newtheorem{theorem}{Theorem}
\newtheorem{corollary}{Corollary}

\newtheorem{lemma}{Lemma}

\def\wl{\par \vspace{\baselineskip}}
\theoremstyle{definition}

\theoremstyle{remark}

\newtheorem{assumption}{Assumption}

\newcommand{\GG}{\mathbb{G}}

\newcommand{\ber}{{\sf Ber}}

\newcommand{\unif}{{\sf Unif}}
\newcommand{\nm}{{\sf N}}

\newcommand{\gam}{{\sf Gamma}}

\newcommand{\RR}{\mathbb{R}}

\newcommand{\PP}{\mathbb{P}}

\renewcommand{\L}{\mathscr{L}}

\newcommand{\Nbrack}{N_{[\,]}}
\newcommand{\Jbrack}{J_{[\,]}}

\newcommand{\sign}{\mathrm{sign}}

\newcommand{\veta}{\boldsymbol{\veta}}

\renewcommand{\phi}{\varphi} 
\newcommand{\eps}{\varepsilon}

\title{Gibbs posterior inference on the minimum clinically important difference}
\author{
Nicholas Syring \quad and \quad Ryan Martin \\
Department of Statistics \\
North Carolina State University \\
{\tt (nasyring, rgmarti3)@ncsu.edu} }
\date{\today}

\begin{document}


\maketitle 

\begin{abstract}
It is known that a statistically significant treatment may not be clinically significant.  A quantity that can be used to assess clinical significance is called the minimum clinically important difference (MCID), and inference on the MCID is an important and challenging problem.  Modeling for the purpose of inference on the MCID is non-trivial, and concerns about bias from a misspecified parametric model or inefficiency from a nonparametric model motivate an alternative approach to balance robustness and efficiency.  
In particular, a recently proposed representation of the MCID as the minimizer of a suitable risk function makes it possible to construct a Gibbs posterior distribution for the MCID without specifying a model.  We establish the posterior convergence rate and show, numerically, that an appropriately scaled version of this Gibbs posterior yields interval estimates for the MCID which are both valid and efficient even for relatively small sample sizes.  

\smallskip

\end{abstract}

\section{Introduction}
\label{S:intro}

In clinical trials, often the main objective is assessing the efficacy of a treatment.  However, experts have observed that statistical significance alone does not necessarily imply efficacy \citep{jacobson.truax.1991}.  For instance, a study with high power can detect statistically significant differences, but these may not translate to practical differences noticeable by the patients.  As a result, a cutoff value different than a statistical critical value is desired that can separate patients with and without clinically significant responses.  This cutoff is called the {\em minimum clinically important difference}, or {\em MCID} for short \citep{jaescheke1989}.  Accurate inference on the MCID is crucial for clinicians and health policy-makers to make educated judgments about the effectiveness of certain treatments.  Indeed, the U.~S.~Food and Drug Administration held a special workshop in 2012 on methodological developments towards improved inference on the MCID.\footnote{\url{https://federalregister.gov/a/2012-27147}}


	

The basic setup is that, in addition to a scalar diagnostic measure for each patient, which would be used to assess the statistical significance of a treatment, one also has access to a ``patient-reported outcome,'' a binary indicator of whether or not the patient felt that the treatment was beneficial.  Then, roughly, the MCID is defined as the cutoff value such that, if the diagnostic measure exceeds this cutoff, then the patient is likely to observe a benefit from the treatment.  A more precise description of the problem setup is given in Section~\ref{S:mcid}.  The challenge in making inference on the MCID is in modeling the joint distribution for the diagnostic measure and patient-reported outcome.  Given a model, standard likelihood-based methods---Bayesian or non-Bayesian---could be used, but specifying a sound model is difficult because the MCID is a rather complicated functional thereof.  To avoid the potential bias caused by a misspecified parametric model and the inefficiencies that result from an overly-complex nonparametric model, a model-free approach is an attractive alternative.  Recently, \citet{xu.mcid} propose a M-estimation framework for estimating the MCID, that does not require a model, but the distribution theory needed to provide valid tests or confidence intervals for the MCID based on their approach is apparently out of reach.  


In this paper, we show that a Gibbs posterior can provide inference on the MCID without requiring a likelihood, thus avoiding the modeling step and the risk of misspecification while providing easy access to credible intervals, and that our new method compares favorably to the existing M-estimation method in terms of both large-sample theory and finite-sample performance.  Construction of the Gibbs posterior takes advantage, first, of the representation in \citet{xu.mcid} of the MCID as the minimizer of an expected loss and, second, of the recent efforts in \citet{bissiri.holmes.walker.2013} to describe a Bayesian-like analysis with a loss function in place of a likelihood.  Our Gibbs posterior distribution is easy to compute and, with a suitable scaling, is shown to provide valid and efficient credible intervals for the MCID.  

Our focus in this paper is the MCID application, but some general comments about Gibbs posteriors are worth mentioning.  First, our problem is related to that of model misspecification, and it is known \citep[e.g.,][]{bunke.milhaud.1998, lee.macheachern.2011, walker2013, kleijn, deblasi.walker.2013, rvr.sriram.martin} that, asymptotically, the posterior distribution behaves reasonably under misspecification provided that it is Gibbs-like in the sense that the negative log-likelihood used resembles a suitable loss function; a nice example of this type is \citet{sriram.rvr.ghosh.2013}.  Second, although misspecification is usually viewed as a bad thing, there might be reasons to ``misspecify on purpose.''  
For example, one may not wish to spend the resources needed to flesh out a full model, including priors, and to compute the full posterior when, ultimately, it will be marginalized to the parameter of interest.  The Gibbs posterior described here has the advantage of being defined directly on the parameter of interest, simplifying both prior specifications and posterior computations.  

The remainder of the paper is organized as follows.  In Section~\ref{S:mcid} we introduce our notation for the MCID problem and formulate its definition as a minimizer of an expected loss.  This leads naturally to the M-estimator proposed in \citet{xu.mcid} and we improve on their asymptotic convergence rate result in two ways: first, we improve the rate and, second, we clarify the sense in which the rate depends on the local properties of the function defined in \eqref{eq:p.function}.  In Section~\ref{S:bayes}, after a motivating illustration, 
we define our Gibbs posterior distribution for the MCID, and we go on to show that it, and the corresponding posterior mean, converge at the same rate as the M-estimator of \citet{xu.mcid}.  Simulation results are presented in Section~\ref{S:examples}, and the take away message is that our Gibbs posterior, or a suitably scaled version thereof, provides quality inference on MCID, in terms of estimation accuracy and interval coverage and length.  Some concluding remarks are given in Section~\ref{S:discuss}, and technical details are given in the Appendix.

\section{Minimum clinically important difference}
\label{S:mcid}

\subsection{Notation and definitions}

In clinical trials for drugs or medical devices, it is standard to judge the effectiveness of the treatment based on statistical significance.  However, it is possible that the treatment effect may be significantly different from zero in a statistical context, but the effect size is so small that the patients do not experience an improvement.  To avoid the costs associated with bringing to market a treatment that is not clinically effective, it is advantageous to bring the patients' assessment of the treatment effect into the analysis.  While the need for a measure of clinical significance is well-documented \citep[e.g.,][]{kaul.diamond.2010}, it seems there is no universal definition of MCID and, consequently, there is no standard methodology to make inference on it.   Recent efforts in this direction were made by \citet{shiu.gatsonis.2008} and \citet{turner.etal.2010}.  \citet{xu.mcid} provide a mathematically convenient formulation, described next, in which the MCID is expressed as a minimizer of a suitable loss function.  

Let $Y \in \{-1, 1\}$ denote the patient reported outcome with ``$Y=1$'' meaning that the treatment was effective and ``$Y=-1$'' meaning that the treatment was not effective.  Let $X$ be a continuous diagnostic measure taken on each patient.  Let $P$ denote the joint distribution of $(X,Y)$, and $p$ the marginal density of $X$ with respect to Lebesgue measure.  Given $\theta \in \RR$, define the function $\ell_\theta$ by 
\begin{equation}
\label{eq:loss_fcn}
 \ell_\theta(x,y) = \tfrac12\{1 - y \, \sign(x-\theta)\}, \quad (x,y) \in \RR \times \{-1,1\}, 
\end{equation}
where $\sign(0)=1$, and write $R(\theta) = P\ell_\theta$ for the risk function, the expectation of $\ell_\theta$ with respect to the joint distribution $P$.  Then the MCID, denoted by $\theta^\star$, is defined as 
\begin{equation}
\label{eq:mcid}
\theta^\star = \arg \min_\theta R(\theta). 
\end{equation}
That is, the MCID is the minimizer of the risk function $R$, and depends on the distribution $P$ in a rather complicated way.
The intuition behind this definition is the alternative expression for $R(\theta)$:
\[ R(\theta) = P\{Y \neq \sign(X - \theta)\}, \]
i.e., $\theta^\star$ minimizes, over $\theta$, the probability that $\sign(X-\theta)$ disagrees with $Y$.  In other words, $\sign(X-\theta^\star)$ is the best predictor of $Y$ in terms of minimum misclassification probability.  Another representation of the MCID, as demonstrated by \citet{xu.mcid}, that will be convenient below is as a solution to the equation $\eta(\theta) = \frac12$, where 
\begin{equation}
\label{eq:p.function}
\eta(x) = P(Y=1 \mid X=x) 
\end{equation}
is the conditional probability function.  If $\eta$ is continuous and strictly increasing, then $\theta^\star$ will be the unique solution to the equation $\eta(\theta)=\frac12$.  If $\eta$ is only upper semi-continuous, then we may define $\theta^\star$ as $\inf\{x: \eta(x) \geq \frac12\}$, and an argument similar to that in Lemma~1 of \citet{xu.mcid} shows that this $\theta^\star$ solves the optimization problem \eqref{eq:mcid}.



\subsection{M-estimator and its large-sample properties}
\label{SS:minimizer}

\citet{xu.mcid} propose to estimate the MCID by minimizing an empirical risk.  Let $\PP_n = n^{-1}\sum_{i=1}^n \delta_{(X_i,Y_i)}$ be the empirical measure, based on the observations $\{(X_i,Y_i):i=1,\ldots,n\}$, where $\delta_{(x,y)}$ is the point-mass measure at $(x,y)$.  Then the empirical risk is $R_n(\theta) = \PP_n \ell_\theta$, and an M-estimator of MCID is obtained by minimizing $R_n(\theta)$, i.e.,
\begin{equation}
\label{eq:m.est}
\hat\theta_n = \arg\min_\theta R_n(\theta).
\end{equation}
Computation of the estimator is straightforward since it takes only finitely many values depending on the order statistics for the $X$-sample.  Therefore, a simple grid search is guaranteed to quickly identify the minimizer $\hat\theta_n$.  

A shortcoming of this approach is that, due to the discontinuity of the loss function, an asymptotic normality result for the M-estimator does not seem possible; see Section~\ref{SS:smooth}.  Therefore, valid confidence intervals for the MCID based on the M-estimator are not currently available.  This provides motivation for a Bayesian approach, where credible intervals, etc, can be easily obtained, but some non-standard ideas are needed to deal with the fact that $\theta$ is defined by a loss function, not a likelihood; see Section~\ref{S:bayes}.  Bootstrap methods are available (see Section~\ref{S:examples}) but there is a general concern about their validity because the rate is not the usual $n^{-1/2}$.   

Consistency and convergence rates for the M-estimator $\hat\theta_n$ have been studied by \citet{xu.mcid}.  The rates rely on the local behavior of the function $\eta$ and of the marginal distribution of $X$ around $\theta^\star$.  In Theorem~\ref{thm:m.rate}, we clarify and substantially improve upon the rate result given in \citet{xu.mcid}.  Our assumptions here are more efficient than theirs, and we discuss these differences below.  

\begin{assumption}
\label{asp:one}
The marginal density $p$ of $X$ is continuous and bounded away from 0 and $\infty$ on an interval containing $\theta^\star$.
\end{assumption}

\begin{assumption}
\label{asp:two}
The function $\eta$ in \eqref{eq:p.function} is non-decreasing, upper semi-continuous, and satisfies $\eta(\theta) > \eta(\theta^\star$) for all $\theta > \theta^\star$.  Furthermore, there exists constants $c>0$, and $\gamma\geq 0$ such that
\begin{equation}
\label{eq:gamma.new}
\min \left|\eta(\theta^\star \pm \eps) - \eta(\theta^\star)\right| > c\eps^{\gamma}, \quad \text{for all small $\eps > 0$}, 
\end{equation} 
where ``min'' is with respect to the two choices in ``$\pm$.''  
\end{assumption}

We interpret $\gamma$ as an``ease of identification'' index, where smaller $\gamma$ means that the $\eta$ function is, in a certain sense, changing more rapidly near $\theta^\star$, making the MCID easier to identify.  In particular, if $\eta$ has a jump discontinuity at $\theta^\star$, then $\gamma=0$, and this corresponds to the easiest case; if $\eta$ is differentiable at $\theta^\star$, then $\gamma=1$, the most difficult case; and if $\eta$ is continuous but not differentiable at $\theta^\star$, then $\gamma \in (0,1)$, an intermediate case.  For a quick example of the latter case, intermediate ease of identification, fix $\alpha, \beta \in (0,1)$, $\alpha \geq \beta$, and define $\eta(x)$, $x \in [-1,1]$ as 
\[ \eta(x) = \begin{cases} \frac12(1-|x|^{\alpha}), & \text{if $x \in [-1,0)$}, \\ \frac12 (1+x^{\beta}), & \text{if $x \in [0,1]$}. \end{cases}. \]
Clearly, the MCID is $\theta^\star=0$, $\eta$ is continuous but not differentiable there, and \eqref{eq:gamma.new} holds with $\gamma = \alpha$.  Although this ``ease of identification'' index is non-standard, it appears to be the key determinant of the convergence rate.  Indeed, the convergence rate of the M-estimator in Theorem~\ref{thm:m.rate} below improves as $\gamma$ decreases to 0, explaining why we call $\gamma=0$ and $\gamma=1$ the ``easiest'' and the ``most difficult'' cases, respectively.  

\begin{theorem}
\label{thm:m.rate}
Under Assumptions~\ref{asp:one}--\ref{asp:two}, the M-estimator $\hat\theta_n$ in \eqref{eq:m.est} satisfies $\hat\theta_n - \theta^\star = O_P(n^{-r})$ as $n \to \infty$, where $r = (1 + 2\gamma)^{-1}$, and $\gamma$ is defined in \eqref{eq:gamma.new}.   
\end{theorem}

\begin{proof}
See Appendix~\ref{proofs:minimizer}.  
\end{proof}

Our assumptions are different than those in the M-estimator convergence rate theorem of \citet{xu.mcid}, so some comments are in order.  In particular, they impose a H\"older continuity condition on $\eta$, as well as a ``low noise assumption,'' in Equation~(4) in their paper, which upper-bounds the $P$-probability assigned to events of the form $\{|\eta(X) - \frac12| \leq \xi\}$.  This implicitly requires that $\eta$ not be too flat near $\theta^\star$, just like our condition \eqref{eq:gamma.new}, and together with their H\"older condition, they derive a locally uniform lower bound on the risk difference $R(\theta) - R(\theta^\star)$, similar to the one we derive in Lemma~\ref{lem:loss.diff.bnd} in the Appendix.  However, our approach and more-direct assumptions appear to be more efficient, because we get a better lower bound on $R(\theta) - R(\theta^\star)$ and, consequently, a better convergence rate.  Indeed, in the case where $\eta$ is differentiable at $\theta^\star$, we obtain a rate $n^{-1/3}$ whereas \citet{xu.mcid} obtains $n^{-1/5}$ (up to logarithmic terms).  Similarly, for that example above with powers $\alpha \geq \beta$, we obtain a rate $n^{-r}$, with $r=(1+2\alpha)^{-1}$ whereas \citet{xu.mcid} obtains $n^{-r'}$, with $r' = \{2(1+2\alpha)-\beta/\alpha\}^{-1}$.  So, besides showing how the rate depends critically on the ``ease of identification'' index $\gamma$, these examples also highlight the significant improvements in our rates.

\subsection{On smoothed versions of the problem}
\label{SS:smooth}

It was mentioned above that $R_n(\theta)$ not being smooth causes some problems in terms of limit distribution theory, etc.  It would, therefore, be tempting to replace that non-smooth loss function by something smooth, and hope that the approximation error is negligible.  One idea would be to introduce a nice parametric model for this problem.  For example, consider a binary regression model, where $\eta(x) = F(\beta_0 + \beta_1 x)$ and $F$ is some specified distribution function, such as logistic or normal.  Then the MCID corresponds to the median lethal dose \citep[e.g.,][]{kelly2001, agresti2002}.  Such a model is smooth so asymptotic normality holds.  However, unless the true $P$ has the specified form, there will be non-zero bias that cannot be overcome, even asymptotically; see Section~\ref{SS:motivate}.  Since the bias is unknown, sampling distribution concentration around the wrong point cannot be corrected, so is of little practical value.  

A slightly less extreme smoothing of the problem is to make a minor adjustment to the original loss function $\ell_\theta$.  As in \citet{xu.mcid}, introduce a smoothing parameter $\tau > 0$ and consider 
\[ \ell_\theta^\tau(x,y) = \min\bigl\{1, \bigl[1 - \tau^{-1} \, y \, \sign(x-\theta) \bigr]^+ \bigr\}, \]
where $u^+ = \max(u, 0)$ denotes the positive part.  Write $R^\tau(\theta) = P \ell_\theta^\tau$.  Based on arguments in \citet{xu.mcid}, it can be shown that $R^\tau(\theta)$ converges uniformly to $R(\theta)$ as $\tau \to 0$, so, for small $\tau$, the minimizer of $R^\tau$ would be close to $\theta^\star$.  For fixed $\tau$, one can define $R_n^\tau(\theta) = \PP_n \ell_\theta^\tau$ just as before and consider an M-estimator $\hat\theta_n^\tau = \arg \min_\theta R_n^\tau(\theta)$.  An asymptotic normality result for $\hat\theta_n^\tau$ is available, but the proper centering is not at $\theta^\star$ and the asymptotic variance is inversely proportional to $\tau$.  So, one could take $\tau=\tau_n$ vanishing with $n$ in an effort to remove the bias, but a price must be paid in terms of the variance.  Again, having an asymptotic normality result with either an unknown non-zero bias or a very large variance is of little practical value.

Based on these remarks, apparently there is no hope in trying to smooth out the problem to make it a standard one with the usual asymptotic distribution theory.  So, in order to construct useful interval estimates, etc, one needs some different ideas.  

\section{A Gibbs posterior for MCID}
\label{S:bayes}

\subsection{Motivation}
\label{SS:motivate}

As discussed above, estimation of the MCID can be achieved without specifying a model, but the distribution theory needed to develop valid interval estimates is lacking.  A Bayesian approach automatically provides uncertainty quantification, but it requires a model for the joint distribution $P$.  To motivate our model-free Gibbs posterior development that follows, we demonstrate the apparent sensitivity of some ``standard'' Bayesian posterior distributions---parametric and nonparametric---to the underlying $P$.  To be clear, we do not claim that Bayesian methods, in general, are inappropriate for this MCID problem, only that the posterior can be particularly sensitive to the choice of model for $P$ so a less-sensitive approach, if one were available, would be attractive.  


Suppose we begin our analysis with a model for $P$ given by a joint density/mass function $f_\beta(x,y)$, depending on some parameter $\beta$, possibly infinite-dimensional, which would typically be different from $\theta$.  Given a prior for $\beta$, a posterior distribution for $\beta$ can be readily obtained via Bayes theorem, which can be marginalized to get a posterior distribution for $\theta$.  In particular, logistic regression is a sort of black-box approach to study the relationship between a binary response and a quantitative predictor, so consider a Bernoulli model for $Y$, given $X=x$, where the success probability is $F(\beta_0 + \beta_1 x)$, where $F$ is the standard logistic distribution function.  In this case, the MCID is just the median lethal dose, i.e., $\theta = -\beta_0 / \beta_1$.  The choice of the logit link function $F$ is quite rigid, but a more flexible nonparametric approach is available \citep{choudhuri.2007}.  

There are pros and cons to both of the approaches just described.  Assuming that the logistic regression model is well-specified, inference on the MCID ought to be efficient.  However, if the model is misspecified in some way, then there could be non-negligible bias that cannot be overcome, even asymptotically.  The model that treats the link function nonparametrically is more flexible and, therefore, less prone to bias, but at the cost of an increased computational burden and lower efficiency, i.e., posterior for the MCID is more diffuse.  Old-fashioned modeling would be a middle-ground between the extremes of a black-box logistic regression and an overly complex nonparametric regression, but this certainly requires some investment and, unfortunately, is not foolproof.  Our proposed Gibbs approach is an alternative middle-ground, one that avoids misspecification bias, computational and statistical inefficiency, and modeling investment. 



For clarity, we give an illustration of the points just raised.  In particular, we compare our Gibbs posterior defined in Section~3.2 to both a standard Bayesian logistic regression and a nonparametric binary regression \citep{choudhuri.2007}.  For the Bayesian logistic regression we consider the vague priors for $(\beta_0, \beta_1)$ given in \citet{polson.2013}, but our results do not appear to be sensitive to this choice. 
Let us suppose that the true model generating data $(X,Y)$ has a distribution function $F$ for $X$ and, given $X=x$, $Y$ is Bernoulli $\pm 1$ with success probability $F(x)$.  We will consider two different forms of $F$, both two-component normal mixtures:
\[ X \sim 0.7 \nm(-1,1) + 0.3 \nm(1,1) \quad \text{and} \quad X \sim 0.7 \nm(-1,1) + 0.3 \nm(3, 1). \]
The true MCID may be calculated by solving \eqref{eq:p.function}; it is equal to the median of the $X$ distribution, specifically $\theta^\star = -0.514$ in the first example and $\theta^\star = -0.434$ in the second example.  Of course, the logistic regression model is misspecified, but the nonparametric model should not be affected by this.  But how will they perform in the two examples?  

Plots of the marginal posterior density for $\theta$ are shown in Figure~\ref{fig:logistic} for a simulated data set of size $n=500$ obtained from each of the three methods---Gibbs, logistic regression, and nonparametric---one for each marginal distribution for $X$.  In Panel~(a) we see that the posterior distributions for all three models put their mass near the true MCID.  However, in Panel~(b) we see that the posterior distribution for the Bayesian logistic regression is clearly biased away from the true MCID.  The nonparametric Bayesian posterior is very spread out, making it less informative for inference on the MCID.  Our Gibbs approach, however, is right on the mark in both cases, suggesting that it is neither sensitive to model misspecification nor does it suffer from the inefficiency of the nonparametric approach.  

\begin{figure}
\begin{center}
\subfigure[$X \sim 0.7 \nm(-1,1) + 0.3 \nm(1,1)$]{\scalebox{0.5}{\includegraphics{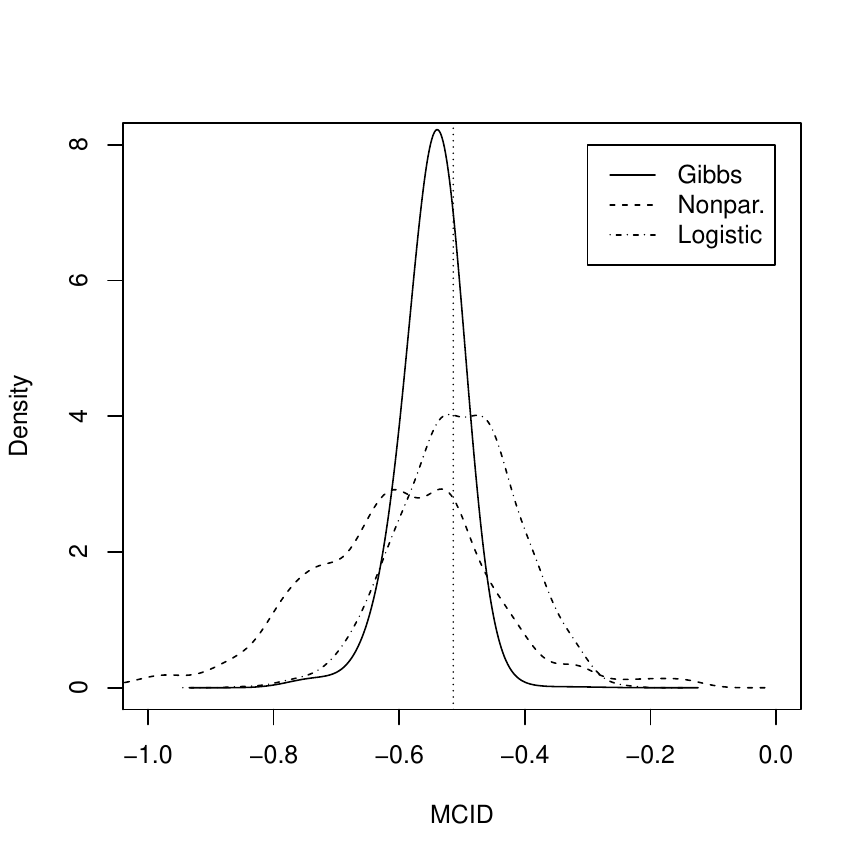}}}
\subfigure[$X \sim 0.7 \nm(-1,1) + 0.3 \nm(3,1)$]{\scalebox{0.5}{\includegraphics{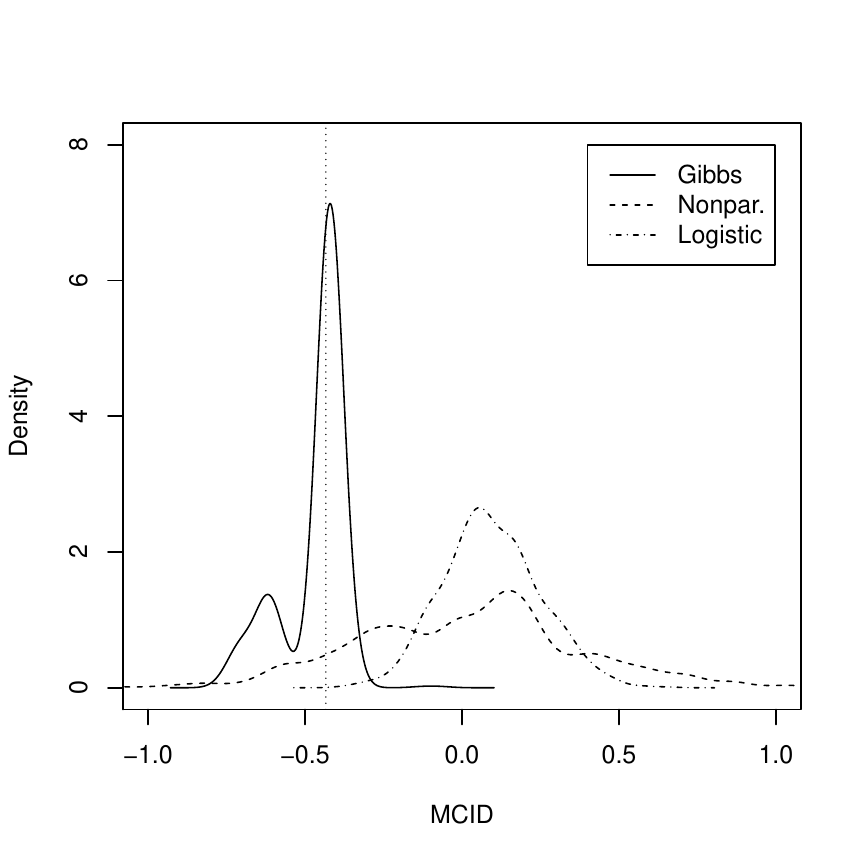}}}
\end{center}
\caption{Plots of (kernel estimates of) the posterior density for MCID.  ``Nonpar." corresponds to the nonparametric binary regression model; ``Logistic'' corresponds to the posterior based on the genuine Bayes logistic model; ``Gibbs'' is the proposed likelihood-free Bayesian posterior; true MCID $\theta^\star$ marked with a dotted vertical line.}
\label{fig:logistic}
\end{figure}

\subsection{Posterior construction}
\label{SS:posterior}

In contrast with a likelihood-based approach, a Gibbs model does not require specification of the probability model $P$.  A Gibbs model consists of a risk function connecting data and parameter, here \eqref{eq:loss_fcn}, and a prior for the parameter.  \citet{bissiri.holmes.walker.2013} consider a Gibbs model that boils down to treating the scaled empirical risk function $nR_n(\theta)$ like a negative log-likelihood and constructing the posterior distribution as usual.  That is, our Gibbs posterior distribution for $\theta$ is given by 
\begin{equation}
\label{eq:post}
\Pi_n(A) = \frac{\int_A e^{-n R_n(\theta)} \,\Pi(d\theta)}{\int_\RR e^{-nR_n(\theta)} \,\Pi(d\theta)}, \quad A \subset \RR, 
\end{equation}
where $R_n(\theta) = \PP_n \ell_\theta$ is the empirical risk defined above, and $\Pi$ is the  prior distribution for $\theta$.  Note that the use of the actual loss function defining the MCID means that we have not introduced any bias.  Moreover, we are only required to do prior specification and posterior computations directly on the $\theta$-space, i.e., there are no additional nuisance parameters that need priors but will ultimately be marginalized away.  

Since the empirical risk function $R_n(\theta)$ is bounded away from zero and infinity, the tails of the posterior match those of the prior.  However, data cannot support a value of $\theta$ outside the range of the $X$ observations so, in practice, we will implicitly restrict the posterior to that range.  This adjustment is not necessary for our theoretical analysis.

\ifthenelse{1=1}{}{
\subsection{Posterior consistency}
\label{SS:consistency}

As a first check that the G posterior is behaving reasonably, we prove an asymptotic consistency theorem.  Roughly, the theorem says that, if $\eta$ in \eqref{eq:p.function} is semi-continuous and strictly increasing, and if the prior density $\pi$ is positive in a neighborhood of $\theta^\star$, then the posterior will, with $P$-probability~1, as $n \to \infty$, put all its mass on arbitrarily small neighborhoods of $\theta^\star$.  Here and in what follows, it will be convenient to rewrite the posterior distribution in \eqref{eq:post} as 
\[ \Pi_n(A) = \frac{N_n(A)}{D_n} = \frac{\int_A e^{-n [L_n(\theta) - L_n(\theta^\star)]} \pi(\theta) \,d\theta}{\int_\RR e^{-n[L_n(\theta) - L_n(\theta^\star)]} \pi(\theta) \,d\theta}. \]
Then the proofs that follow are all based on obtaining appropriate bounds on the numerator, $N_n(A)$, for various subsets $A$, and the denominator, $D_n$.  

\begin{theorem}
\label{thm:consistency}
Let $\eta$ in \eqref{eq:p.function} be semi-continuous and strictly increasing.  If the prior density $\pi$ is positive and continuous in a neighborhood of $\theta^\star$, then $\Pi_n(\{\theta: |\theta-\theta^\star| > \eps\}) \to 0$ $P$-almost surely as $n \to \infty$ for any $\eps > 0$.  
\end{theorem} 

\begin{proof}
See Appendix~\ref{proofs:consistency}.  
\end{proof}

An immediate consequence is that certain estimators based on the posterior distribution, such as the posterior mean, are consistent. 

\begin{corollary}
\label{crl:post.mean}
If the prior mean for $\theta$ exists, then, under the conditions of Theorem~\ref{thm:consistency}, the posterior mean $\tilde\theta_n$ satisfies $|\tilde\theta_n - \theta^\star| \to 0$ $P$-almost surely as $n \to \infty$.  
\end{corollary}

\begin{proof}
See Appendix~\ref{proofs:consistency}.  
\end{proof}
}

\subsection{Posterior convergence rates}
\label{SS:rates}

The Gibbs posterior convergence rate describes, roughly, the size of the neighborhood around $\theta^\star$ that it assigns nearly all its mass as $n \to \infty$.  An important consequence of a posterior convergence rate result is that typical posterior summaries also have nice convergence rate properties; see Corollary~\ref{crl:post.mean.rate}.  

It turns out that the posterior convergence rate result holds under virtually the same conditions as Theorem~\ref{thm:m.rate} for the M-estimator.  The only additional condition needed concerns the prior, and it is very mild.  

\begin{assumption}
\label{asp:three}
The prior distribution $\Pi$ for $\theta$ has a density $\pi$ which is continuous and bounded away from zero in a neighborhood of $\theta^\star$.  
\end{assumption}

\begin{theorem}
\label{thm:post.rate}
Under Assumptions~\ref{asp:one}--\ref{asp:three}, the Gibbs posterior distribution $\Pi_n$ in \eqref{eq:post} satisfies $\Pi_n(A_n) = o_P(1)$ as $n \to \infty$, where $A_n = \{\theta: |\theta-\theta^\star| > a_n n^{-r}\}$, $r = (1+2\gamma)^{-1}$ for $\gamma$ in \eqref{eq:gamma.new}, and $a_n$ is any diverging sequence.  
\end{theorem}

\begin{proof}
See Appendix~\ref{proofs:rates}.  
\end{proof}


\begin{corollary}
\label{crl:post.mean.rate}
Under the conditions of Theorem~\ref{thm:post.rate}, if the prior mean for $\theta$ exists, then the posterior mean $\tilde\theta_n$ satisfies $\tilde\theta_n - \theta^\star = O_P(n^{-r})$ as $n \to \infty$.  
\end{corollary}

\begin{proof}
See Appendix~\ref{proofs:rates}. 
\end{proof}

\subsection{On scaling the loss function}
\label{SS:scale}

A subtle point is that the loss function $\ell_\theta$ has an arbitrary scale.  That is, the problem of inference on the risk minimizer is unchanged if we replace $\ell_\theta$ with $\omega \ell_\theta$ for any $\omega > 0$.  While this has no effect on the M-estimator, it does have an effect on our Gibbs posterior.  Based on our experience, the posterior distribution actually tends to be quite narrow, so the posterior concentration rates seem to be driven primarily by the ``center'' of the posterior and, to a lesser extent, by the ``spread.''  As \citet{bissiri.holmes.walker.2013} explain, it is important to scale the loss in some way.  In Lemma~\ref{lem:scale} below, we show that the Gibbs posterior may be scaled by a vanishing sequence without sacrificing the convergence rate in Theorem~\ref{thm:post.rate}.  Our numerical results in Section~\ref{S:examples} show that taking the scale to be vanishing accomplishes the goal of calibrating the credible intervals without affecting the accuracy of the posterior mean estimator.  

\begin{lemma}
\label{lem:scale}
Under Assumptions~\ref{asp:one}--\ref{asp:three}, with $r = (1+2\gamma)^{-1}$ and $\gamma$ defined in \eqref{eq:gamma.new}, the conclusion of Theorem~\ref{thm:post.rate} holds if the loss function $\ell_\theta$ is scaled by a sequence $\omega_n$ that vanishes strictly more slowly than $n^{-\gamma r}$.  
\end{lemma}

\begin{proof}
Similar to the proof of Theorem~\ref{thm:post.rate} in Appendix~\ref{proofs:rates}. 
\end{proof}

In our experience, with continuous $\eta$, we found that a scale value of approximately $\omega_n = c n^{-1/4}$, for $c \in (1,2)$, worked well in terms of credible interval calibration.  To avoid making an ad hoc choice of constant, we employ the algorithm in \citet{syring.martin.gibbs}.  Our scaling algorithm is applied to each simulated data set, producing a different, data-dependent value of the scale parameter each time.  Briefly, $\omega_n$  is determined by solving the equation that sets the Gibbs posterior credible interval coverage probability equal to the desired confidence level.  The algorithm utilizes standard techniques including stochastic approximation, MCMC, and bootstrapping.  In our simulations, the algorithm succeeds in producing approximately calibrated credible intervals.  In the numerical examples that follow, the $\omega_n$ selected by the algorithm is, on average, roughly $1.5 n^{-1/4}$, which is consistent with the result in Lemma~\ref{lem:scale}.  

\section{Numerical examples}
\label{S:examples}

We consider four examples to illustrate the performance of our Gibbs posterior for the MCID.  Each example has a different marginal distribution for $X$:
\begin{description}
\item[\it Example~1.] $X \sim 0.7 \nm(-1,1) + 0.3 \nm(1,1)$;
\vspace{-2mm}
\item[\it Example~2.] $X \sim \nm(1,1)$; 
\vspace{-2mm}
\item[\it Example~3.] $X \sim \unif(-2,4)$.
\vspace{-2mm}
\item[\it Example~4.] $X \sim \gam(2,0.5)$.
\end{description} 
These examples cover a variety of distributions: bimodal, normal, flat, and skewed.  In each example, we take $n$ independent samples from the respective marginal distributions, and then, given $X_i=x_i$, take $Y_i$ as a $\pm 1$ Bernoulli with probability $F(x_i)$, $i=1,\ldots,n$, where $F$ is the distribution function of $X$ and $\ber(p)$ denotes a Bernoulli distribution with success probability $p$.  In our case, the relevant summaries are the bias and standard deviation of the estimators, and the coverage probability and length of the 90\% interval estimates.  We considered three sample sizes, namely, $n=250, 500, 1000$, and the results in Tables~\ref{tab:estimates}--\ref{tab:intervals} are based on 1000 Monte Carlo samples.  We compare the performance of our Gibbs posterior, using the scaling algorithm in \citet{syring.martin.gibbs} and a flat prior for $\theta$, to a baseline method, namely, the M-estimator and the corresponding percentile bootstrap confidence intervals.  

Table~\ref{tab:estimates} shows the empirical bias and standard deviation for both the M-estimator and the Gibbs posterior mean while Table~\ref{tab:intervals} shows the empirical coverage probability and length for the 90\% interval estimates based on bootstrapping the M-estimator and on the Gibbs posterior sample.  Here we see that the additional flexibility of being able to choose the scaling parameter/sequence provides approximately calibrated posterior credible intervals for each $n$.  Overall, the performance of our Gibbs posterior is comparable to the M-estimator+bootstrap, the take-away message being that a Bayesian-like approach need not sacrifice desirable frequentist properties.  Two comments are in order.  First, if reliable prior information is available, which is possible in medical applications where studies are replicated, then this can be readily incorporated into our analysis, naturally providing some improvements.  For example, if an accurate, informative $\nm(-0.5,1)$ prior is used in Example~1 for $n = 250$, the bias is reduced to $0.01$ and the credible interval length is reduced to $0.83$ with $0.90$ coverage, an improvement over bootstrap confidence intervals.  Second, the desirable frequentist properties are not automatic for other Bayesian approaches; for instance, the Bayesian logistic regression model described in Section~\ref{SS:motivate} has mean square error equal to $0.07$ and $0.20$ in Examples~3 and 4, respectively, with $n=250$, compared to $0.015$ and $0.011$ for the Gibbs posterior mean.

\begin{table}[t]
\centering
\begin{tabular}{clcccc}
Example & Method & $n=250$ & $n=500$ & $n=1000$ \\ 
\hline
1 & M-estimator  & $0.03$ (0.21) & $0.01$ (0.16) & $0.01$ (0.12) \\
          & Gibbs & $0.03$ (0.22) & $0.01$ (0.17) & $0.01$ (0.13) \\ 
          & & & & \\
2 & M-estimator & $0.02$ (0.16) & $0.02$ (0.12) & $0.01$ (0.10) \\
          & Gibbs & $0.00$ (0.15) & $0.01$ (0.12) & $0.00$ (0.10) \\ 
          & & & & \\
3 & M-estimator & $0.01$ (0.12)&$0.01$ (0.10)& $0.01$ (0.08) \\
          & Gibbs &$0.01$ (0.12)& $0.00$ (0.10) & $0.00$ (0.07) \\ 
					& & & & \\
4 & M-estimator & $0.01$ (0.12) & $0.01$ (0.09) & $0.01$ (0.07) \\
          & Gibbs & $0.03$ (0.10) & $0.02$ (0.08) & $0.01$ (0.06) \\ 
\hline
\end{tabular}
\caption{Absolute empirical bias (and standard deviation) for the M-estimator in \citet{xu.mcid} and our proposed Gibbs posterior mean.}
\label{tab:estimates}
\wl
\centering
\begin{tabular}{clcccc}
Example & Method & $n=250$ & $n=500$ & $n=1000$ \\ 
\hline
1 & M+Boot  & $0.91$ (0.86) & $0.91$ (0.69) & $0.93$ (0.53) \\
          & Gibbs & $0.89$ (0.89) & $0.89$ (0.69) & $0.91$ (0.55) \\ 
          & & & & \\
2 & M+Boot &$0.91$ (0.60) & $0.91$ (0.48) & $0.92$ (0.38) \\
          & Gibbs & $0.89$ (0.61) & $0.91$ (0.50) & $0.90$ (0.38) \\ 
          & & & & \\
3 & M+Boot & $0.90$ (0.47) & $0.90$ (0.38) & $0.91$ (0.30) \\
          & Gibbs & $0.91$ (0.48) & $0.90$ (0.37) & $0.90$ (0.30) \\ 
					& & & & \\
4 & M+Boot & $0.92$ (0.39) & $0.92$ (0.31) & $0.92$ (0.25) \\
          & Gibbs & $0.91$ (0.41) & $0.90$ (0.31) & $0.90$ (0.24) \\ 
\hline
\end{tabular}
\caption{Empirical coverage probability (and mean length) of $90\%$ confidence interval based on bootstrapping the M-estimator in \citet{xu.mcid} and the 90\% credible interval from our proposed Gibbs posterior.}
\label{tab:intervals}
\end{table}

\section{Conclusion}
\label{S:discuss}

In this paper, motivated by a real application in medical statistics, we have explored the use of a Gibbs posterior distribution for inference.  In certain applications, like this MCID problem, the statistician may be reluctant to use a likelihood-based model due to fear of misspecification, computational difficulty, or for some other reason.  The Gibbs model offers an alternative Bayesian-like approach that does not require a probability model, thus avoiding some of these potential challenges, and we think that this advantage may make Gibbs models widely applicable.  As we have demonstrated, the proposed Gibbs posterior is theoretically justified and provides quality point and interval estimates in practice.  So, in a certain sense, our Gibbs posterior provides the best of both worlds: that is, we get a theoretically justifiable posterior distribution without the unnecessary modeling and computations and without worry of model misspecification.  

The technical details in this paper are kept relatively simple due to the fact that $\theta$ is a scalar and $\ell_\theta$ is bounded, but our methods can be applied more generally.  For example, \citet{xu.mcid} proposed a generalization of the MCID problem in which $\theta$ is actually a function of some other covariates, thereby making the MCID ``personalized'' in a certain sense.  We are working on extending both the theory and the computational methods presented here to this more general case.  A recent paper \citep{syring.martin.image} develops a nonparametric Gibbs posterior for inference on a function and prove a smoothness-adaptive convergence rate theorem.  The techniques developed therein may be able to be applied in the personalized MCID problem.

\section*{Acknowledgments}

The authors are grateful to the Editor, Associate Editor, and referees for their helpful comments on a previous version of this manuscript.  This work is partially supported by the U.~S.~Army Research Offices, Award \#W911NF-15-1-0154.

\appendix

\section{Technical details and proofs}
\label{S:proofs}

\subsection{Preliminary results}
\label{SS:prelim}

Here, for the sake of completeness, we summarize some basic facts about the empirical risk function $R_n(\theta) = \PP_n \ell_\theta$ and the risk $R(\theta) = P \ell_\theta$.  Details will be given only for those results not taken directly from \citet{xu.mcid}.  

First, we consider properties of the expected loss difference, $R(\theta)-R(\theta^\star)$.  By definition of $\theta^\star$, and Assumption~\ref{asp:two} about $\eta$, we know the difference is strictly positive except at $\theta=\theta^\star$.  To see this, \citet{xu.mcid} show that 
\begin{equation}
\label{eq:loss.diff}
R(\theta) - R(\theta^\star) = 2 \int_{\theta^\star}^\theta \{\eta(x) - \tfrac12\} p(x) \,dx. 
\end{equation}
Moreover, by continuity of $p$ in Assumption~\ref{asp:one} and almost everywhere continuity of $\eta$ derived from Assumption~\ref{asp:two}, we can see that the derivative of $R(\theta)-R(\theta^\star)$ {could be} zero only at $\theta=\theta^\star$, which implies that the function is uniformly bounded away from zero outside an interval containing $\theta^\star$.  This latter point is important because asymptotic results of, say, the M-estimator require that the minimizer $\theta^\star$ be ``well-separated'' \citep[e.g.,][Theorem~5.7]{vaart1998}.  We seek a lower bound on the expected loss difference in \eqref{eq:loss.diff} for parameter values far from the MCID.  That is, we want to calculate
\begin{equation}
\label{eq:inf}
\inf_{|\theta - \theta^\star|>\delta} R(\theta) - R(\theta^\star).
\end{equation}
The following result is new and allows us to improve upon the rate given for the M-estimator in \citet{xu.mcid}.

\begin{lemma}
\label{lem:loss.diff.bnd}
Under Assumptions~\ref{asp:one}--\ref{asp:two}, there exists a constant $c > 0$ such that \eqref{eq:inf} is lower-bounded by $c\delta^{1+\gamma}$ for all sufficiently small $\delta>0$.   
\end{lemma}

\begin{proof}
Since $\eta(x)$ is non-decreasing in $x$ (Assumption~\ref{asp:two}) the infimum in \eqref{eq:inf} occurs at the boundary, either at $\theta^\star + \delta$ or at $\theta^\star - \delta$.  The two cases can be handled similarly, so we give the argument only for the case that the infimum is attained at $\theta^\star + \delta$.  Monotonicity of $\eta$ implies that $\eta(\theta^\star + \delta) > \eta(\theta^\star + \delta/2) > \eta(\theta^\star)$.  Also, according to Assumption~\ref{asp:one}, the marginal density $p$ is bounded away from zero on an interval containing $\theta^\star$, so let $b$ be the infimum over the interval $[\theta^\star - \delta, \theta^\star + \delta]$. Using the expression in \eqref{eq:loss.diff}, we can lower-bound $R(\theta^\star+\delta) - R(\theta^\star)$ as follows:
\begin{align*}
R(\theta^\star + \delta) - R(\theta^\star) & = \int_{\theta^\star}^{\theta^\star + \delta} \{2 \eta(x) - 1\} p(x) \,dx \\
& = \Bigl(\int_{\theta^\star}^{\theta^\star + \delta/2} + \int_{\theta^\star + \delta/2}^{\theta^\star + \delta}  \Bigr) \{2\eta(x) - 1\} p(x) \, dx \\
& > b \, \delta \, \{\eta(\theta^\star + \delta / 2) - \tfrac12\} \\
& \geq b \, \delta \, \{\eta(\theta^\star + \delta/2) - \eta(\theta^\star)\}.
\end{align*}
By Assumption~\ref{asp:two}, in particular, condition \eqref{eq:gamma.new}, we have that the difference in the last display is bounded below by $c_1 (\delta / 2)^{\gamma}$.  Plugging this in at the end of the above display gives the advertised lower bound, $c \delta^{1 + \gamma}$, where $c=b c_1 / 2^{\gamma}$.  
\end{proof}

\ifthenelse{1=1}{}{
\[ \int_{\theta^\star}^{\theta^\star+\delta} (2\eta(x)-1)p(x)dx > \int_{\theta^\star}^{\theta^\star+\xi} (2\eta(x)-1)p(x)dx + c_p(\delta - \xi) [2\eta(\theta^\star + \xi)-1].\]
Since the integral on the right hand side is positive, we can ignore it and simplify the bound to
\[\int_{\theta^\star}^{\theta^\star+\delta} (2\eta(x)-1)p(x)dx > c_p(\delta - \xi) [2\eta(\theta^\star + \xi)-1].\]
Make the choice $\xi = \delta/2$ and use the bounds in Assumption~\ref{asp:two} to express the bound on the risk difference as
\begin{align*}
\int_{\theta^\star}^{\theta=\theta^\star+\delta} (2\eta(x)-1)p(x)dx &> c_p\delta/2 [2\eta(\theta^\star + \delta/2)-1] \\
&> c_p\delta [\eta(\theta^\star + \delta/2)-1/2]\\
&> c_pc_1\delta^{1+\gamma}
\end{align*}
using that $\eta(\theta^\star) = 1/2$.  

Next, suppose that the infimum occurs at $\theta^\star - \delta$ instead.  Since $\theta^\star$ is defined as $\inf\{x:\eta(x)\geq 1/2\}$, by Assumption~\ref{asp:two} there exists a $\xi>0$ where $\theta^\star - \delta < \theta^\star - \xi < \theta^\star$ such that $\eta(\theta^\star - \delta) \leq \eta(\theta^\star - \xi) < \eta(\theta^\star)$.  Then, by following the above argument, we have that
\begin{align*}
\int_{\theta^\star-\delta}^{\theta^\star} (1-2\eta(x))p(x)dx &> c_p\delta/2 [1 - 2\eta(\theta^\star - \delta/2)] \\
&> c_p\delta [1/2 - \eta(\theta^\star - \delta/2)]\\
&> c_pc_2\delta^{1+\gamma}.
\end{align*} 
In each case, the infimum can be bounded by an expression of the form $c\delta^{1+\gamma}$, as was to be shown. 
}

Second, we need some approximation properties of the class of functions 
\[ \L_\delta := \{\ell_\theta-\ell_{\theta^\star}: |\theta-\theta^\star| < \delta\}, \quad \delta > 0. \]
\citet{xu.mcid} shows, using the standard partition in the classical Glivenko--Cantelli theorem \citep[e.g.,][Example~19.6]{vaart1998}, that the $L_1(P)$ $\eps$-bracketing number $\Nbrack(\eps, \L_\infty, L_1(P))$ is proportional to $\eps^{-1}$.  This is enough to show that the class $\L_\infty$ is Glivenko--Cantelli, from which a uniform law of large numbers follows.  However, better rates can be obtained by using a local bracketing, i.e., of $\L_\delta$ for finite $\delta$, and Assumptions~\ref{asp:one}--\ref{asp:two}.  Such considerations allow us to remove the unnecessary logarithmic term on the rate presented in Theorem~1 of \citet{xu.mcid}.  

\begin{lemma}
\label{lem:bracket}
$\Nbrack(\eps, \L_\delta, L_1(P)) \lesssim \delta/\eps$.  
\end{lemma}

\begin{proof}
For the standard Glivenko--Cantelli theorem partition, which is used in \citet{xu.mcid}, one needs to partition the interval $[0,1]$ into $k$ intervals of length less than $\eps$, so $k$ must be greater than $1/\eps$, but can be taken less than $2/\eps$.  By Assumption~\ref{asp:one}, we have that $\delta \lesssim P(|X-\theta^\star| < \delta) \lesssim \delta$.  For the local bracketing, this means we only need to partition an interval of length proportional to $\delta$ into intervals of length less than $\eps$.  Therefore, the total number of intervals is $\lesssim \delta/\eps$, as was to be shown.  
\end{proof}

From this and the fact that the brackets are pairs of indicator functions, we can get a bound on the $L_2(P)$ bracket number, i.e., $\Nbrack(\eps, \L_\delta, L_2(P)) \lesssim (\delta/\eps)^2$; see Example~19.6 in \citet{vaart1998}.  Then the bracketing integral is 
\begin{equation}
\label{eq:jbrack}
\Jbrack(\delta, \L_\delta, L_2(P)) := \int_0^\delta \{\log \Nbrack(\eps, \L_\delta, L_2(P)) \}^{1/2} \,d\eps \lesssim \delta. 
\end{equation}

Finally, we will need a maximal inequality for the empirical process $\GG_n(\ell_\theta-\theta_{\theta^\star})$ for $\theta$ near $\theta^\star$.  \citet{xu.mcid} show that $g(\theta) = I_{\{|\theta-\theta^\star| \leq \delta\}}$ is an envelop function for $\L_\delta$, with $\|g\|_{L_2(P)} \lesssim \delta^{1/2}$.  Then, given the bound \eqref{eq:jbrack} on the bracketing integral, the maximal inequality in Corollary~19.35 of \citet{vaart1998} gives the following.  

\begin{lemma}
\label{lem:maximal}
$P\{\sup_{|\theta-\theta^\star| < \delta} |\GG_n(\ell_\theta-\ell_{\theta^\star})|\} \lesssim \delta^{1/2}$. 
\end{lemma}

\subsection{Proofs from Section~\ref{SS:minimizer}}
\label{proofs:minimizer}

\begin{proof}[Proof of Theorem~\ref{thm:m.rate}]
Similar to the proof of Theorem~2 in \citet{wongshen1995} and of Theorem~5.52 in \citet{vaart1998}.  The M-estimator $\hat\theta_n$, the global minimizer of $R_n$, satisfies $R_n(\hat\theta_n) \leq R_n(\theta^\star) + \zeta_n$ for \emph{any} $\zeta_n$.  Let $K>0$ be as in Lemma~\ref{lem:num} in Appendix~\ref{proofs:rates}, and take $\zeta_n = Ks_n^{1+\gamma}$, where $s_n = a_nn^{-r}$, $r=(1+2\gamma)^{-1}$, and $a_n$ is any divergent sequence.  Then we have that 
\[ |\hat\theta_n - \theta^\star| > s_n \implies \sup_{|\theta-\theta^\star| > s_n} \{R_n(\theta^\star) - R_n(\theta) \} \geq -K s_n^{1+\gamma}. \]
By Lemma~\ref{lem:num}, the latter event has vanishing $P$-probability, which implies that $P(|\hat\theta-\theta^\star| > s_n) \to 0$.  Therefore, $\hat\theta_n - \theta^\star = o_P(s_n)$ or, since $a_n$ is arbitrary, $\hat\theta_n - \theta^\star = O_P(n^{-r})$.  
\end{proof}

\ifthenelse{1=1}{}{
\subsection{Proofs from Section~\ref{SS:consistency}}
\label{proofs:consistency}

\begin{proof}[Proof of Theorem~\ref{thm:consistency}]
The denominator $D_n$ of the posterior distribution can be bounded below by $e^{-nc}$ for any $c > 0$, with $P$-probability~1 for all large $n$.  The proof of this is exactly like that of Lemma~4.4.1 in \citet{ghoshramamoorthi}.  For the numerator, we make use of the uniform law of large numbers in Lemma~\ref{lem:ulln} in Appendix~\ref{SS:prelim}.  For given $\eps > 0$, let $A=\{\theta: |\theta-\theta^\star| > \eps\}$.  Then 
\begin{align*}
N_n(A) & = \int_A e^{-n[L_n(\theta) - L_n(\theta^\star)]} \pi(\theta) \,d\theta \\
& \leq e^{n^{1/2} \sup_\theta |\GG_n(\ell_\theta - \ell_{\theta^\star})|} \int_A e^{-n[L(\theta) - L(\theta^\star)]} \pi(\theta) \,d\theta. 
\end{align*}
By the uniform law of large numbers, the first term is bounded above by $e^{n^{1/2} b}$ for any $b > 0$, $P$-almost surely for all large $n$.  Also, because $L(\theta)-L(\theta^\star)$ is positive and bowl-shaped, for $\theta \in A$, there exists a constant $a > 0$ such that $L(\theta) - L(\theta^\star) \geq a$ for $\theta \in A$.  Therefore, 
\[ N_n(A) \leq e^{n^{1/2} b - na}, \quad \text{$P$-almost surely for all large $n$}. \]
Then there exists $d > 0$ such that $N_n(A) \leq e^{-nd}$ $P$-almost surely for all large $n$.  Now combine the bounds on $N_n(A)$ and $D_n$, with $c=d/2$, to complete the proof.  
\end{proof}

\begin{proof}[Proof of Corollary~\ref{crl:post.mean}]
By Jensen's inequality, $|\tilde\theta_n - \theta^\star| \leq \int |\theta - \theta^\star| \,\Pi_n(d\theta)$.  Take any $\eps > 0$ and partition $\RR$ as $\{\theta: |\theta-\theta^\star| \leq \eps\} \cup \{\theta: |\theta-\theta^\star| > \eps\}$ and write 
\[ \int |\theta - \theta^\star| \, \Pi_n(d\theta) \leq \eps + \int_{|\theta-\theta^\star| > \eps} |\theta-\theta^\star| \,\Pi_n(d\theta). \]
Just like in the proof of Theorem~\ref{thm:consistency}, the posterior measure away from $\theta^\star$ can be bounded by the prior measure times some $Z_n$ such that $Z_n \to 0$ $P$-almost surely, uniformly in $\theta$.  Therefore, the second term in the above display is bounded by $Z_n \cdot \int |\theta-\theta^\star| \, \Pi(d\theta)$, which itself vanishes $P$-almost surely since the prior mean exists.  Since $\eps > 0$ is arbitrary, the claimed consistency follows.  
\end{proof}
}

\subsection{Proofs from Section~\ref{SS:rates}}
\label{proofs:rates}

Here, it will be convenient to rewrite the posterior distribution in \eqref{eq:post} as 
\[ \Pi_n(A) = \frac{N_n(A)}{D_n} = \frac{\int_A e^{-n [R_n(\theta) - R_n(\theta^\star)]} \pi(\theta) \,d\theta}{\int_\RR e^{-n[R_n(\theta) - R_n(\theta^\star)]} \pi(\theta) \,d\theta}. \]
Then the goal is to obtain appropriate upper bounds on the numerator and lower bounds on the denominator.  For the latter, we have the following result, whose proof follows that of Lemma~1 in \citet{shen.wasserman.2001} almost exactly; the only difference is that boundedness of $\ell_\theta-\ell_{\theta^\star}$ can be used in place of a second Kullback--Leibler moment.    

\begin{lemma}
\label{lem:den}
For a vanishing sequence $t_n$, set $\Theta_n = \{\theta: R(\theta)-R(\theta^\star) \leq t_n\}$.  If $nt_n \to \infty$, then $D_n \gtrsim \Pi(\Theta_n) e^{-2nt_n}$ with $P$-probability converging to 1 as $n \to \infty$.
\end{lemma}

The next step is to bound the numerator $N_n(A_n)$, where $A_n$ is the complement of the shrinking neighborhood of $\theta^\star$ define in Theorem~\ref{thm:post.rate}.  Towards this, we have the following technical result which provides some uniform control on the empirical risk difference $R_n(\theta^\star)-R_n(\theta)$ for $\theta$ outside a neighborhood of $\theta^\star$.  The proof of this result relies on the bracketing entropy calculations and maximal inequality in Appendix~\ref{SS:prelim}.   

\begin{lemma}
\label{lem:num}
Under Assumptions~\ref{asp:one}--\ref{asp:two}, with $\gamma$ defined in \eqref{eq:gamma.new}, let $s_n = a_n n^{-r}$ where $r= (1+2\gamma)^{-1}$, and $a_n$ is any diverging sequence.  Then there exists $K > 0$ such that 
\[ P\Bigl(\sup_{|\theta-\theta^\star| > s_n} \{R_n(\theta^\star) - R_n(\theta) \} > -K s_n^{1+\gamma} \Bigr) \to 0, \quad \text{as $n \to \infty$}. \]
\end{lemma}

\begin{proof}
Start with the identity 
\[ R_n(\theta^\star) - R_n(\theta) = \{ R(\theta^\star) - R(\theta) \} - n^{-1/2} \GG_n(\ell_\theta - \ell_{\theta^\star}),  \]
where $\GG_n f = n^{1/2}( \PP_n f - Pf)$ is the empirical process.  Next, since the supremum of a sum is no more than the sum of the suprema, we get 
\[ \sup_{|\theta-\theta^\star| > \eps} \{ R_n(\theta^\star) - R_n(\theta) \} \leq \sup_{|\theta-\theta^\star| > \eps} \{ R(\theta^\star) - R(\theta) \} + n^{-1/2} \sup_{|\theta-\theta^\star| > \eps} |\GG_n(\ell_\theta - \ell_{\theta^\star})|; \]
the second inequality comes from putting absolute value on the empirical process term.  From Lemma~\ref{lem:loss.diff.bnd}, we get 
\[ \sup_{|\theta-\theta^\star| > \eps} \{ R_n(\theta^\star) - R_n(\theta) \} \leq -C \eps^{1+\gamma} + n^{-1/2} \sup_{|\theta-\theta^\star| > \eps} |\GG_n(\ell_\theta - \ell_{\theta^\star})|. \]
Now, following the proof of Theorem~5.52 from \citet{vaart1998} or of Theorem~1 in \citet{wongshen1995}, introduce ``shells'' $\{\theta: 2^m \eps < |\theta-\theta^\star| \leq 2^{m+1} \eps\}$ for integers $m$.  On these shells, we can use both the bound in Lemma~\ref{lem:loss.diff.bnd} and the maximal inequality in Lemma~\ref{lem:maximal}.  That is, 
\begin{align*}
\sup_{|\theta-\theta^\star| > s_n} & \{R_n(\theta^\star) - R_n(\theta)\} > -K s_n^{1+\gamma} \\
& \implies \sup_{2^m s_n < |\theta-\theta^\star| \leq 2^{m+1} s_n} \{R_n(\theta^\star) - R_n(\theta)\} > -K s_n^{1+\gamma} \quad \exists \; m \geq 0 \\
& \implies n^{-1/2} \sup_{2^m s_n < |\theta-\theta^\star| <  2^{m+1} s_n} |\GG_n(\ell_\theta-\ell_{\theta^\star})| \geq C (2^m s_n)^{1+\gamma} - K s_n^{1+\gamma} \\
& \implies n^{-1/2} \sup_{|\theta-\theta^\star| \leq  2^{m+1} s_n} |\GG_n(\ell_\theta-\ell_{\theta^\star})| \geq C (2^m s_n)^{1+\gamma} - K s_n^{(1+\gamma)},
\end{align*}
If $K \leq C/2$, then $C (2^m s)^{(1+\gamma)} - K s^{1+\gamma} \geq C(2^m s)^{1+\gamma} / 2$ for all $m \geq 0$.   
\begin{align*}
P\Bigl( \sup_{|\theta-\theta^\star| > s_n} & \{R_n(\theta^\star) - R_n(\theta)\} > -K s_n^{1+\gamma} \Bigr) \\
& \leq \sum_{m \geq 0} P\Bigl(  n^{-1/2} \sup_{|\theta-\theta^\star| <  2^{m+1} s_n} |\GG_n(\ell_\theta-\ell_{\theta^\star})| \geq C (2^m s_n)^{1+\gamma} / 2 \Bigr)
\end{align*}
To the summands, apply Markov's inequality and Lemma~\ref{lem:maximal} to get 
\[ P\Bigl(  n^{-1/2} \sup_{|\theta-\theta^\star| <  2^{m+1} s_n} |\GG_n(\ell_\theta-\ell_{\theta^\star})| \geq C (2^m s_n)^{1+\gamma} / 2 \Bigr) \leq \frac{C' (2^{m+1} s_n)^{1 / 2}}{n^{1/2} (2^m s_n)^{1+\gamma}}. \]
For $s_n = a_n n^{-1/(1+2\gamma)}$ where $a_n \uparrow \infty$, the upper bound satisfies
\[ \lesssim 2^{1/2}2^{-m(1/2+\gamma)}a_n^{-1/2-\gamma}. \]
Since $\gamma > 0$, the sum over $m \geq 0$ converges, so
\[ P\Bigl( \sup_{|\theta-\theta^\star| > s_n} \{R_n(\theta^\star) - R_n(\theta)\} > -K s_n^{1+\gamma} \Bigr) \lesssim a_n^{-1/2-\gamma}. \]
Then the upper bound vanishes since $a_n \uparrow \infty$ and $-1/2-\gamma<0$, completing the proof.  
\end{proof}

\begin{proof}[Proof of Theorem \ref{thm:post.rate}]
From Lemma~\ref{lem:num} we get an exponential bound on the numerator $N_n(A_n)$, i.e., $N_n(A_n) \leq \exp\{-K n s_n^{1+\gamma}\}$ with $P$-probability approaching 1.  For the denominator $D_n$, for a suitable sequence $t_n$, we have $D_n \gtrsim \Pi(\Theta_n) e^{-2n t_n}$, where $\Theta_n = \{\theta: R(\theta) - R(\theta^\star) \leq t_n\}$.  We claim that 
\[ \Theta_n \supseteq \{\theta: |\theta-\theta^\star| \leq t_n\}. \]
To see this, first note that if $\theta$ is close to $\theta^\star$, then by an argument similar to that in the proof of Lemma~\ref{lem:loss.diff.bnd} and using the boundedness of $\eta$,
\[ R(\theta) - R(\theta^\star) \lesssim |\theta-\theta^\star|^{\gamma_2} \int_{\theta^\star}^\theta p(x) \,dx. \]
The remaining term in the upper bound is the marginal $P$-probability assigned to the small interval around $\theta^\star$ which, by Assumption~\ref{asp:one}, can be bounded by a constant times $|\theta-\theta^\star|$.  Therefore, $R(\theta)-R(\theta^\star) \lesssim |\theta-\theta^\star|$ so, if $|\theta-\theta^\star| \lesssim t_n$, then $R(\theta)-R(\theta^\star) \leq t_n$.  Under Assumption~\ref{asp:three}, we can bound $\Pi(\Theta_n) \gtrsim t_n$.  So, if we take $t_n = n^{-(1-\beta)}$ for some $\beta > 0$ to be identified and some constant $H>0$, then by Lemma~\ref{lem:den} we get 
\[ D_n \gtrsim e^{-H n^\beta}, \quad \text{with $P$-probability approaching 1}. \]
Putting together the bounds on the numerator and denominator we get that, for some constant $M$, 
\[ \frac{N_n(A_n)}{D_n} \lesssim \exp\Bigl\{-M \Bigl(a_n^{1+\gamma} n^{\frac{\gamma_1}{1+2\gamma}} - n^\beta \Bigr) \Bigr\}. \]
We can take $\beta < \frac{\gamma}{1+2\gamma}$, and the upper bound vanishes.  
\end{proof}

\begin{proof}[Proof of Corollary~\ref{crl:post.mean.rate}]
Set $s_n = a_n n^{-r}$ for $a_n$ an arbitrary divergent sequence.  Next, define $\tilde s_n = \tilde a_n n^{-r(\gamma)}$, where $\tilde a_n$ is such that $\tilde a_n / a_n \to 0$, e.g., $\tilde a_n = \log a_n$.  Now partition $\RR$ as $\{\theta: |\theta-\theta^\star| \leq \tilde s_n\} \cup \{\theta: |\theta-\theta^\star| > \tilde s_n\}$, and write 
\begin{equation}
\label{eq:post.mean.rate.bound}
|\tilde\theta_n - \theta^\star| \leq \int |\theta - \theta^\star| \, \Pi_n(d\theta) \leq \tilde s_n + \int_{|\theta-\theta^\star| > \tilde s_n} |\theta-\theta^\star| \,\Pi_n(d\theta), 
\end{equation}
where the first inequality is by Jensen.  From the proof of Theorem~\ref{thm:post.rate}, the posterior away from $\theta^\star$ is bounded by the prior times some $Z_n = o_P(1)$, uniformly in $\theta$.  That is, 
\[ \int_{|\theta-\theta^\star| > \tilde s_n} |\theta-\theta^\star| \, \Pi_n(d\theta) \leq Z_n \int |\theta-\theta^\star| \, \Pi(d\theta). \]
In fact, we can bound $Z_n$ more precisely:
\[ Z_n \lesssim \exp\Bigl\{-M \Bigl(\tilde{a}_n^{1+\gamma} n^{\frac{\gamma}{1+2\gamma}} - n^\beta \Bigr) \Bigr\}, \quad \text{sufficiently small $\beta > 0$}. \] 
Dividing through \eqref{eq:post.mean.rate.bound} by $s_n$ we get that $s_n^{-1} |\tilde\theta_n - \theta^\star|$ is bounded by a constant times 
\[ \tilde a_n / a_n + e^{-\zeta_n} \int |\theta-\theta^\star| \, \Pi(d\theta), \]
where $\zeta_n = M \tilde{a}_n^{1+\gamma} n^{\frac{\gamma}{1+2\gamma}} - Mn^\beta +\log a_n  - r \log n$.  The first term in the upper bound goes to zero by the choice of $\tilde a_n$.  The second term goes to zero provided that the prior mean exists and $\zeta_n \to \infty$ as $n \to \infty$.  We assumed the former condition, and the latter can be easily arranged by choosing $\beta$ sufficiently small, so $\tilde\theta_n - \theta^\star = o_P(s_n)$.
\end{proof}

\bibliographystyle{apalike}
\bibliography{mybib}

\end{document}